\documentclass[aps,onecolumn,letterpaper,superscriptaddress,nofootinbib,notitlepage]{revtex4}

\usepackage{lipsum}
\usepackage{setspace}

\usepackage{graphicx}
\usepackage{amsmath}
\usepackage{amssymb}
\usepackage{amsthm}
\usepackage[utf8]{inputenc}
\usepackage{comment}

\usepackage{color}

\newcommand{\norm}[1]{\left\lVert#1\right\rVert} 

\usepackage[english]{babel}
\usepackage[T1]{fontenc}
\usepackage{mathrsfs}
\usepackage{braket}
\usepackage{physics}
\usepackage{enumerate}
\usepackage{graphicx}
\usepackage{mathtools}
\usepackage[export]{adjustbox}
\usepackage{microtype}

\theoremstyle{plain}

\usepackage[colorlinks = true,
            linkcolor = red,
            urlcolor  = blue,
            citecolor = blue,
            anchorcolor = red]{hyperref}

\newtheorem{theorem}{Theorem}

\newtheorem{proposition}[theorem]{Proposition}
\newtheorem{lemma}[theorem]{Lemma}
 
\newtheorem{definition}[theorem]{Definition}  
\newtheorem{remark}[theorem]{Remark}  
\newtheorem*{remark*}{Remark}   
\renewcommand\qedsymbol{$\blacksquare$}
\newenvironment{proof-of}[1][{\hspace{-\blank}}]{{\medskip\noindent\textit{Proof~{#1}.\ }}}{\hfill\qedsymbol}

\renewcommand{\Tr}{{\operatorname{Tr}\,}}

\newcommand{\cid}{{\text{id}}}

\newcommand{\1}{\openone}
\newcommand{\proj}[1]{|#1\rangle\!\langle #1|}

\newcommand{\cT}{{\mathcal{T}}}

\newcommand{\KI}{{\mathrm{KI}}}

\newcommand{\nc}{\newcommand}
\nc{\rnc}{\renewcommand}
\nc{\avg}[1]{\langle#1\rangle}
\nc{\Rank}{\operatorname{Rank}}
\nc{\smfrac}[2]{\mbox{$\frac{#1}{#2}$}}

\nc{\ox}{\otimes}
\nc{\dg}{\dagger}
\nc{\dn}{\downarrow}
\nc{\cA}{{\cal A}}
\nc{\cB}{{\cal B}}
\nc{\cC}{{\cal C}}
\nc{\cF}{{\cal F}}
\nc{\cG}{{\cal G}}
\nc{\cH}{{\cal H}}
\nc{\cI}{{\cal I}}
\nc{\cJ}{{\cal J}}
\nc{\cK}{{\cal K}}
\nc{\cL}{{\cal L}}
\nc{\cM}{{\cal M}}
\nc{\cN}{{\cal N}}
\nc{\cO}{{\cal O}}
\nc{\cP}{{\cal P}}
\nc{\cQ}{{\cal Q}}
\nc{\cR}{{\cal R}}
\nc{\cS}{{\cal S}}
\nc{\cX}{{\cal X}}
\nc{\cY}{{\cal Y}}
\nc{\cZ}{{\cal Z}}
\nc{\cso}{{\text{o}}}

\nc{\csupp}{{\operatorname{csupp}}}
\nc{\qsupp}{{\operatorname{qsupp}}}
\nc{\rar}{\rightarrow}
\nc{\lrar}{\longrightarrow}
\nc{\polylog}{{\operatorname{polylog}}}
\nc{\wt}{{\operatorname{wt}}}

\nc{\RR}{{{\mathbb R}}}
\nc{\CC}{{{\mathbb C}}}
\nc{\FF}{{{\mathbb F}}}
\nc{\NN}{{{\mathbb N}}}
\nc{\ZZ}{{{\mathbb Z}}}
\nc{\PP}{{{\mathbb P}}}
\nc{\QQ}{{{\mathbb Q}}}
\nc{\UU}{{{\mathbb U}}}
\nc{\EE}{{{\mathbb E}}}

\nc{\Hom}[2]{\mbox{Hom}(\CC^{#1},\CC^{#2})}
\nc{\rU}{\mbox{U}}

\nc{\ob}[1]{#1}

\begin{document}

\title{\huge{Quantum Reverse Shannon Theorem Revisited}}


\author{Zahra Baghali Khanian}
\affiliation{
Technical University of Munich, 85748 Garching, Germany}
\affiliation{Perimeter Institute for Theoretical Physics, Ontario, Canada, N2L 2Y5}
\affiliation{Institute for Quantum Computing, University of Waterloo, Ontario, Canada, N2L 3G1}

\author{Debbie Leung}
\affiliation{Institute for Quantum Computing, University of Waterloo, Ontario, Canada, N2L 3G1}
\affiliation{Perimeter Institute for Theoretical Physics, Ontario, Canada, N2L 2Y5}

\begin{abstract}

Reverse Shannon theorems concern the use of
noiseless channels to simulate noisy ones. This is dual to the usual 
noisy channel coding problem, where a noisy (classical or quantum) channel is used to simulate a noiseless one. 
The Quantum Reverse Shannon Theorem 
is extensively studied by Bennett and co-authors in [IEEE Trans. Inf. Theory, 2014].
%
They present two distinct theorems, each tailored to classical and quantum channel simulations respectively, explaining the fact that these theorems remain incomparable due to the fundamentally different nature of correlations they address. The authors leave as an open question the challenge of formulating a unified theorem that could encompass the principles of both and unify them.
We unify these two theorems into a single, comprehensive theorem, extending it to the most general case by considering correlations with a general mixed-state reference system. Furthermore, we unify feedback and non-feedback theorems by simulating a general side information system at the encoder side.
%


\end{abstract}
\maketitle


\section{Introduction}


The classical ``reverse Shannon theorem'' was established and proven in 2002 in \cite{Bennett1999} as a dual to Shannon's capacity theorem. 
This theorem states that for any channel $\cN$ with capacity $C$, if the sender and receiver share an unlimited supply of random bits, then an expected $Cn + o(n)$ uses of a noiseless binary channel are sufficient to \textit{simulate} $n$ uses of the channel $\cN$.
%
%
%
%
The essence of this theorem is that, in the presence of shared randomness, the asymptotic properties of a classical channel can be characterized by a single parameter: its capacity.
%
A quantum generalization of the reverse Shannon theorem is formulated and extensively studied  
in \cite{Bennett2014a}. 
They consider shared entanglement as the quantum counterpart of shared randomness and obtain the optimal quantum 
{simulation} rates under different structures and available resources: free entanglement, restricted entanglement, tensor power input states, arbitrary input states, and feedback and non-feedback simulation models.  (Additional study with different techniques can be found in \cite{BCR2011}.)

One of the questions that remained open in \cite{Bennett2014a} is the different treatment of the classical and quantum cases.  
In this paper, we address this problem by considering tensor power mixed input states shared between the encoder and a reference system.  
This not only unifies the classical and quantum models but also extends them to the most general quantum case.
We also unify the coherent feedback and non-feedback models and extend it to the most general case
by preserving an arbitrary system at the encoder side. 
In the presence of free entanglement, we fully characterize the optimal simulation rate in terms of a quantity that resembles 
the entanglement-assisted capacity \cite{Bennett1999}.  
Considering the general mixed-state case comes with its own complications, as properties used in analyzing pure quantum states, such as the monogamy of entanglement, are not applicable to mixed states.
Without the assistance of entanglement, we obtain converse and achievability bounds, which involve similar quantities but
differ in the limit taken for the error.
It is not obvious whether these bounds match in general, but we provide various examples for which the two bounds are equal.

 We introduce two functionals $a(\rho,\gamma)$ and $u(\rho,\gamma)$ of a quantum state $\rho$ and an error $\gamma$.
 The first functional has properties such as sub-additivity and continuity, and it fully characterizes the assisted simulation rates. 
The second functional is more complex, and it characterizes the simulation rate in the unassisted model.  
Even for partially classical input states, it can evaluate to the entanglement of purification, which is not known to be additive.
Hence, even without the  issues in the limit of the error, the rate is multi-letter
and hard to compute.

\bigskip


The structure of the paper is as follows: At the end of this section, we briefly introduce the notation used in this paper.
In Section~\ref{sec:setup} we rigorously define the channel simulation model.
We discuss a decoupling lemma in Sec~\ref{sec: decoupling and functionals}, and introduce two functionals, which characterize the simulation rates.
We obtain the optimal simulation rates assuming that the parties share
free entanglement and no entanglement  in Sec~\ref{sec: E-assisted} and  Sec~\ref{sec: unassisted}, respectively.
We discuss our results in Sec~\ref{sec: discussion}.
In the Appendix, we introduce and prove some lemmas that we apply throughout the paper.

\bigskip

 \noindent \textbf{Notation.} 
 In this paper, quantum systems are associated with finite dimensional Hilbert spaces $A$, $R$, etc.,
 whose dimensions are denoted by $|A|$, $|R|$, respectively. 
 The von Neumann entropy is defined as 
 \begin{align}
     S(\rho) &:= - \Tr\rho\log\rho. \nonumber
 \end{align}
 Throughout this paper, $\log$ denotes by default the binary logarithm.
 The conditional entropy and the conditional mutual information, $S(A|B)_{\rho}$ and $I(A:B|C)_{\rho}$, respectively, are defined in the same way as their classical counterparts: 
 \begin{align*}
 S(A|B)_{\rho}   &= S(AB)_\rho-S(B)_{\rho}, \text{ and} \\ 
    I(A:B|C)_{\rho} &= S(A|C)_\rho-S(A|BC)_{\rho} \\
                   &= S(AC)_\rho+S(BC)_\rho-S(ABC)_\rho-S(C)_\rho.
 \end{align*}
 The fidelity between two states $\rho$ and $\xi$ is defined as 
 \(
  F(\rho, \xi) = \|\sqrt{\rho}\sqrt{\xi}\|_1 
                 = \Tr \sqrt{\rho^{\frac{1}{2}} \xi \rho^{\frac{1}{2}}},
 \) 
 with $\| \cdot \|_1$ as the Schatten 1-norm.  $\|X\|_1 = \Tr|X| = \Tr\sqrt{X^\dagger X}$. It relates to the trace distance in the following well-known way \cite{Fuchs1999}:
 \begin{equation}
   1-F(\rho,\xi) \leq \frac12\|\rho-\xi\|_1 \leq \sqrt{1-F(\rho,\xi)^2}.
 \end{equation}

\section{Setup}\label{sec:setup}

We assume that an arbitrary channel $\cN:A \to BK$ is given with all the associated dimensions specified, along with a state $\rho^{AR}$ on the input and some reference system $A$ and $R$.  Let $\sigma^{BKR}=(\cN \ox \cid^R)\rho^{AR}$, and $U_{\cN}:A\to BKG$ be the Stinespring dilation of $\cN$.
We consider $n$ copies of the state $\rho^{AR}$. 

We call the sender or the encoder Alice, and the receiver
or the decoder Bob. 
We suppose that Alice and Bob initially share some entangled state $\ket{\Phi}^{A_0 B_0}$
in systems $A_0 B_0$.
Alice applies an encoding channel $\mathcal{C}_n^{A^nA_0 \to MK^n A_1}$, 
and sends system $M$ to Bob. 
Receiving $M$, Bob applies a decoding channel $\mathcal{D}_n^{MB_0\to B^n B_1}$. 
We define 
\begin{align*}
    \nu_n^{MK^nR^nA_1B_0} &\coloneqq (\mathcal{C}_n^{A^nA_0 \to MK^n A_1}\otimes \cid^{R^nB_0})(\rho^{A^nR^n}\otimes \proj{\Phi}^{A_0B_0}),\\ 
    \xi_n^{B^nK^nR^n A_1 B_1} &\coloneqq (\mathcal{D}_n^{MB_0\to B^n B_1}\otimes \cid^{K^nR^nA_1})(\nu_n^{MK^nR^n A_1 B_0}).
\end{align*}
Furthermore, consider the Stinespring dilations $U^{A^nA_0 \to MK^nW_A A_1}_{\mathcal{C}_n}$ and $U^{MB_0\to B^nW_B B_1}_{\mathcal{D}_n}$ for the encoding and the decoding maps. 
We consider the following purifications for the states $\nu_n$ and $\xi_n$, 
\begin{equation*}
    \begin{aligned}
    \ket{\nu_n}^{MK^n W_A R^n {R'}^n A_1 B_0}
    &\coloneqq  (U^{A^nA_0 \to MK^nW_A A_1}_{\mathcal{C}_n}\otimes \1^{ R^n {R'}^n B_0}) (\ket{\rho}^{A^nR^n{R'}^n} \otimes \ket{\Phi}^{A_0 B_0}),
\\[3ex]      
    \ket{\xi_n}^{B^nK^nW_AW_BR^n{R'}^n A_1 B_1} 
    &\coloneqq  (U^{MB_0 \to B^nW_B B_1}_{\mathcal{D}_n}\otimes \1^{K^nW_AR^n{R'}^n A_1})\ket{\nu_n}^{MK^nW_AR^n{R'}^n A_1 B_0}. 
    \end{aligned}
\end{equation*}
We say that the scheme has fidelity $1-\epsilon$ if 
\begin{align}\label{eq: F}
    \cF_n:= F\left(\sigma^{B^nK^nR^n}\otimes \proj{\Phi}^{A_1B_1} , \xi_n^{B^nK^nR^n A_1 B_1}\right)
    \geq 1-\epsilon.
\end{align}  
For a given $(n,\epsilon)$, we define the minimal  qubit  and entanglement rates as
\begin{align}
\mathcal{Q}(n,\epsilon )&:=\frac{1}{n} \log |M| \\
\mathcal{E}(n,\epsilon )&:=\frac{1}{n} (S(A_0)-S(A_1)) 
\end{align}
such that there exists an $(n,\epsilon)$ code with $|M|$ and $S(A_0)-S(A_1)$.
%
We say that a qubit $Q$ and entanglement rate $E$  are asymptotically achievable
if there exists a sequence of codes $\{(\mathcal{C}_n,\mathcal{D}_n)\}_n$ such that 
\begin{align*}
    \cF_n \geq 1-\epsilon_n ~~~{\rm and}~~~ 
   \mathcal{Q}(n,\epsilon ) \leq Q + \delta_n  ~~~{\rm and}~~~  \mathcal{E}(n,\epsilon ) \leq E + \eta_n,
\end{align*}
for some vanishing non-negative sequences $\{ \epsilon_n \}$, $\{ \delta_n \}$ and $\{ \eta_n \}$.
The optimal qubit and entanglement rates are defined respectively as 
\begin{equation*}
     Q^* = \inf\{Q : Q \mathrm{\,\,is\,\,achievable}\},~~ 
     E^* = \inf\{E : E \mathrm{\,\,is\,\,achievable}\}. 
\end{equation*}


Two distinct notions of feedback is introduced in \cite{Bennett2014a} as passive feedback and coherent feedback.
In a passive feedback model, the encoder 
obtains a copy of the decoder's output.
For quantum channels, it is not possible to give the encoder a copy
of the decoder's output because of the no-cloning theorem. 
A coherent feedback  of a channel is defined as an isometry in which the part of the output that does
not go to the decoder is retained by the encoder, rather than escaping to
the environment. 
Classical and coherent feedback are thus rather different
notions. 
\begin{remark}
We  consider coherent feedback in this paper, and refer to it simply as feedback.
Our model unifies the (coherent) feedback and non-feedback simulation of \cite{Bennett2014a} in a single model as follows. 
If we assume that the source state $\rho^{AR}$ is pure, and we let system $K=E$ or $K=\emptyset$ then we recover the feedback or non-feedback channel simulation of \cite{Bennett2014a}, respectively. 
\end{remark}

\section{Decoupling Condition and Rate Functionals}\label{sec: decoupling and functionals}

The fidelity criterion of Eq.~(\ref{eq: F}) implies the decoupling lemma below, where we show that the distilled entanglement systems $A_1B_1$ are decoupled from the rest of the systems. We apply this lemma in our converse proofs.
The proof of this lemma is presented in Sec.~\ref{sec: Proof of decoupling lemma} of the Appendix. 
\begin{lemma}\label{lemma: decoupling}
 The fidelity criterion of Eq.~(\ref{eq: F}) implies that the decoded state on systems 
 $A_1B_1$ is decoupled from the rest of the systems in the following sense
\begin{align}
    I(B^n K^n W_A W_B R^n {R'}^n :A_1 B_1)_{\xi_n} \leq n \delta(n,\epsilon), \nonumber
\end{align}
where $\delta(n,\epsilon) = 4\sqrt{6\epsilon} \log(d_1) + \frac{2}{n}h(\sqrt{6\epsilon})$, 
with $d_1=\frac{\log|A_1|}{n}$ and the binary entropy $h(\epsilon)=-\epsilon\log\epsilon - (1-\epsilon)\log(1-\epsilon)$.
The mutual information is with respect to the decoded state $\xi_n$.
%
\end{lemma}

In Definition~\ref{def: a(rho,epsilon)} described below, 
we define two functions of a state $\rho^{AR}$. 
Our main results are that, these functions characterize the optimal simulation rates.  
Theorem~\ref{thm: E-assisted} of the manuscript states that 
the optimal entanglement-assisted rate for the simulation of the channel
$\cN:A \to BK$ is equal to $a(\rho^{AR},0)$.
Theorem~\ref{thm: unassisted} states that the  regularized rate
$ \lim_{m \to \infty}  \frac{1}{m} u(\rho^{\ox m}, \frac{1}{m^9})$
is achievable.  
Moreover, any achievable quantum rate is lower bounded as
$ \lim_{\gamma \to 0} \lim_{m \to \infty} \frac{1}{m}u(\rho^{\ox m},\gamma)$.


\begin{definition}\label{def: a(rho,epsilon)}
For $\gamma \geq 0$, a state $\rho^{AR}$ and a CPTP map $\cN:A \to BK$ define
    \begin{align}
        a(\rho, \gamma)&:= \inf_{\Lambda_1:A\to BK} \; \frac{1}{2} I(B:RR')_{\tau_1}
        \quad \quad \quad \text{s.t.} \quad F(\sigma^{BKR},\tau_1^{BKR}) \geq 1-\gamma,\nonumber\\
        u(\rho, \gamma)&:= 
        \inf_{\Lambda_3:E  \to E'} \;
        \inf_{\Lambda_2:A \to BK}
        S(BE')_{\tau_3}
        \quad \>\> \> \text{s.t.} \quad F(\sigma^{BKR},\tau_2^{BKR}) \geq 1-\gamma,\nonumber  
    \end{align}
    where $\sigma^{BKR}=(\cN \ox \cid_R)\rho^{AR}$,  
    the maps $\Lambda_{1,2,3}$ are CPTP, and 
    $\Lambda_1:A\to BK$, $\Lambda_2:A\to BK$,  
    $U_{\Lambda_2}: A \hookrightarrow BKE$ is an isometric extension of  $\Lambda_2$, with $E$ as an environment system,     
    $\Lambda_3:E\to E'$ where   
    the choice of $E'$ is part of the optimization, and  
 the states in the above quantities are defined as 
\begin{align*}
\tau_1^{BKRR'}&:=(\Lambda_1\ox \cid_{RR'}) \left(\proj{\rho}^{ARR'} \right)\\
\tau_2^{BKR}&:=(\Lambda_2\ox \cid_R) \left( \proj{\rho}^{ARR'}  \right)\\
\tau_3^{BKE'R}&:=(\Lambda_3  \ox \cid_{BKR}) \left( (U_{\Lambda_2}\ox \1_{R})\proj{\rho}^{ARR'} (U_{\Lambda_2}\ox \1_{R})^{\dagger} \right),
\end{align*} 
with the state $\ket{\rho}^{ARR'}$  a purification of $\rho^{AR}$ and  
    $\tau_1^{BKR}=\Tr_{R'}\tau^{BKRR'}$. 
\end{definition}
 These functions are  defined for a given channel $\cN$, however, we drop the dependency on the channel for the simplicity of the notation. 

\begin{remark}\label{remark: min vs inf}
The infimums in the above definition are attainable, and therefore they can be replaced with  minimums. 
The first optimization is over a compact set of CPTP maps with bounded input and output dimensions. 
In the second optimization, system $E$
is an environment system of the map $\Lambda_2$, which is bounded as 
$|E|\leq |A|\cdot|B|\cdot|K|$.
Also, the von Neumann entropy is a concave function of states. Therefore, the infimum is attained
    by an extremal CPTP map $\Lambda_3:E \to E'$.  The input dimension of $\Lambda_3$ is bounded,
    therefore, the number of the operators in the Kraus representation of  an extremal $\Lambda_3$ with input dimension $E$,
is  $|E|$.  This implies that the dimension of system $E'$ is bounded as well.
\end{remark}

\section{Entanglement-assisted simulation}\label{sec: E-assisted}
In this section, we obtain that the optimal entanglement-assisted qubit rate is equal to $a(\rho^{AR},0)$ (where ``$a$'' stands for the assisted rate). So, we first prove various properties of this function, which we apply to obtain the optimal rate. 
The entanglement-assisted qubit rate means that we allow the encoder and the decoder to consume entanglement at any rate.
\begin{lemma}\label{lemma: f properties}
    The function $a(\rho, \gamma)$  in Definition~\ref{def: a(rho,epsilon)}
    has the following properties:
    \begin{enumerate}
        \item It is a non-increasing function of $\gamma$.
        \item It is convex in $\gamma$, i.e., 
        $a(\rho, \lambda \gamma_1 + (1-\lambda) \gamma_2) \leq 
        \lambda a(\rho, \gamma_1) + (1-\lambda) a(\rho, \gamma_2)$.
        \item It is subadditive, i.e. $a(\rho_1 \ox \rho_2, \gamma)\geq a(\rho_1, \gamma)+a(\rho_2, \gamma)$.
        \item It is continuous for all $\gamma \geq 0$.
    \end{enumerate}
\end{lemma}
We prove this lemma in the appendix  section~\ref{sec: Proof of Lemma f properties}.

\begin{theorem}\label{thm: E-assisted}
    The optimal entanglement-assisted rate for the simulation of the channel
    $\cN:A \to BK$ is equal to $a(\rho^{AR},0)$ where
    this function is defined in Definition~\ref{def: a(rho,epsilon)}.
\end{theorem}

\begin{proof}
The proof of the direct part (achievability of the rate) is as follows.
Let $\Lambda_1:A\to BK$ be the optimal CPTP map in Definition~\ref{def: a(rho,epsilon)}
at $\gamma=0$. Let $U_{\Lambda_1}:A\to BKE$ be the corresponding Stinespring dilation isometry of $\Lambda_1:A\to BK$. Alice applies $U_{\Lambda_1}$ to each copy of the
state $\rho^{AR}$. Then, the overall purified state is
\begin{align}
    \ket{\tau_1}^{BKERR'}=(U_{\Lambda_1} \ox \1_{RR'})\ket{\rho}^{ARR'}, \nonumber 
\end{align}
where $\ket{\rho}^{ARR'}$ is a purification of $\rho^{AR}$, and 
$R$ and $R'$ are inaccessible reference systems. Note that by definition 
$\Tr_{ER'}{\tau_1}^{BKERR'}={\tau}^{BKR}$.
Then Alice and Bob apply QSR to $n$ copies of the pure state $\tau_1^{BKERR'}$ to send system $B^n$ from Alice to Bob with systems $K^nE^n$ as the side information systems of Alice. The rate of this 
protocol is equal to $a(\rho,0)=\frac{1}{2}I(B:RR')_{\tau_1}+\eta_n$. After implementing this protocol, the state shared by Alice and Bob is $\epsilon_n$ close to $(\ket{\tau_1}^{BKERR'})^{\ox n}$.
Tracing out systems $E^nR'^n$ only increases the fidelity, hence, this protocol 
achieves the rate of $a(\rho,0)=\frac{1}{2}I(B:RR')_{\tau_1}+\eta_n$ and preserves the 
$1-\epsilon_n$ fidelity with the state $({\tau}^{BKR})^{\ox n}$. By Theorem~\ref{thm: QSR}, $\eta_n$ and $\epsilon_n$ vanish 
as $n$ grows very large.  
%


\bigskip

In the following, we obtain the converse bound.  For any protocol with block length $n$ and error $\epsilon$ 
\begin{align}\label{eq: unassisted converse}
    S(M)_{\nu}+S(B_0)_{\Phi}&\geq S(MB_0)_{\nu} \nonumber \\
    &= S(B^n W_BB_1)_{\xi} \nonumber \\
    &\geq S(B^n W_B)_{\xi}+S(B_1)_{\xi}-n\delta(n,\epsilon) \nonumber \\
    &\geq S(B^n W_B)_{\xi}+S(B_1)_{\Phi}-n\delta(n,\epsilon)-n\delta_1(n,\epsilon),  
\end{align}
where the second line is due to applying the decoding isometry.
The third line follows from Lemma~\ref{lemma: decoupling}.
The last line follows from the decodability: the 
output state on system $B_1$ is $2\sqrt{2\epsilon}$-close 
to the original state $\Phi$ in trace norm; then the inequality follows 
by applying the Fannes-Audenaert inequality, where 
$\delta_1(n,\epsilon)=\frac1n \sqrt{2\epsilon} \log(|A_1|) + \frac1n h(\sqrt{2\epsilon})$.
From the above, we obtain
\begin{align}\label{Q+E}
    n\mathcal{Q}(n,\epsilon )+n\mathcal{E}(n,\epsilon )&=S(M)_{\nu}+S(B_0)_{\Phi}-S(B_1)_{\Phi} \nonumber\\
     &\geq S(B^n W_B)_{\xi}-n\delta(n,\epsilon)-n\delta_1(n,\epsilon), 
\end{align}
where $n\mathcal{E}(n,\epsilon )=S(B_0)_{\Phi}-S(B_1)_{\Phi}$.
Moreover, we obtain the following
\begin{align}
    S(M)_{\nu}& \geq S(M|K^n W_A A_1)_{\nu} \nonumber \\
    & = S(MK^n W_A A_1)_{\nu}-S(K^n W_A A_1)_{\nu} \nonumber \\
    & = S(A^n A_0)_{\rho\ox\Phi}-S(K^n W_A A_1)_{\nu} \nonumber \\
    & = S(A^n)_{\rho}+S(A_0)_{\Phi}-S(K^n W_A A_1)_{\nu} \nonumber \\
    & = S(R^n {R'}^n)_{\rho}+S(A_0)_{\Phi}-S(K^n W_A A_1)_{\nu} \nonumber \\
    & = S(R^n {R'}^n)_{\rho}+S(A_0)_{\Phi}-S(B^n W_B R^n {R'}^n B_1)_{\xi} \nonumber \\
    & \geq S(R^n {R'}^n)_{\rho}+S(A_0)_{\Phi}-S(B^n W_B R^n {R'}^n )_{\xi}-S(B_1)_{\xi} \nonumber \\
    & \geq S(R^n {R'}^n)_{\rho}+S(A_0)_{\Phi}-S(B^n W_B R^n {R'}^n )_{\xi}-S(B_1)_{\Phi} -n\delta_1(n,\epsilon), \nonumber 
\end{align}
where the third line is due to the definition of the encoding isometry. 
The fifth and sixth lines follow since the states $\ket{\rho}^{ARR'}$ and 
$\ket{\xi_{n}}^{B^n K^n W_A W_B R^n {R'}^n A_1 B_1}$ are pure.
%
The penultimate line follows from subadditivity of entropy.  
The last line follows from the decodability: the 
output state on system $B_1$ is $2\sqrt{2\epsilon}$-close 
to the original state $\Phi$ in trace norm; then the inequality follows 
by applying the Fannes-Audenaert inequality, where 
$\delta_1(n,\epsilon)=\frac1n \sqrt{2\epsilon} \log(|A_1|) + \frac1n h(\sqrt{2\epsilon})$.
In the last line, note that  $S(B_1)_{\Phi}=S(A_1)_{\Phi}$ holds. From the above, we obtain
\begin{align}\label{Q-E}
    n\mathcal{Q}(n,\epsilon )-n\mathcal{E}(n,\epsilon )&= S(M)_{\nu}+S(B_1)_{\Phi}-S(A_0)_{\Phi} \nonumber\\
    & \geq S(R^n {R'}^n)_{\rho}-S(B^n W_B R^n {R'}^n )_{\xi}-n\delta_1(n,\epsilon),
\end{align}
where $n\mathcal{E}(n,\epsilon )=S(A_0)_{\Phi} -S(A_1)_{\Phi} $.
By adding Eq.~(\ref{Q+E}) and Eq.~(\ref{Q-E}), we obtain
\begin{align}
    2n\mathcal{Q}(n,\epsilon ) &\geq S(R^n {R'}^n)_{\rho}-S(B^n W_B R^n {R'}^n )_{\xi}+S(B^n W_B)_{\xi} -n\delta(n,\epsilon)-2n\delta_1(n,\epsilon)  \nonumber\\
    &= I(B^n W_B: R^n {R'}^n )_{\xi}-n\delta(n,\epsilon)-2n\delta_1(n,\epsilon)  \label{Q Bound W_B} \\
    &\geq I(B^n: R^n {R'}^n )_{\xi}-n\delta(n,\epsilon)-2n\delta_1(n,\epsilon), \nonumber\\
    &\geq 2 na(\rho,\epsilon)-n\delta(n,\epsilon)-2n\delta_1(n,\epsilon), \label{Q Bound}
\end{align}
where the third line follows from the data processing inequality. 
The last line follows from 
the definition of $a(\cdot,\epsilon)$ and its superadditivity  
Lemma~\ref{lemma: f properties}.
%
Dividing by $2n$, 
$Q^* \geq a(\rho,\epsilon)-{1 \over 2} \delta(n,\epsilon)-\delta_1(n,\epsilon)$.  
Taking the limit $n\to \infty$ and $\epsilon \to 0$ in either order, 
$\delta(n,\epsilon)+\delta_1(n,\epsilon) \to 0$, so 
\begin{align}\label{Q Bound}
    Q^* &\geq  \lim_{\epsilon \to 0 }a(\rho,\epsilon) \nonumber \\
       &=  a(\rho,0). \nonumber 
\end{align}
The last line follows from Lemma~\ref{lemma: f properties} point 4, i.e., the continuity of the function at $\epsilon=0$.
\end{proof}

The entanglement-assisted simulation of an identity channel was already studied in \cite{general_mixed_state_compression}, where the optimal rate was found to be 
$S(CQ)_{\omega}-\frac{1}{2}S(C)_{\omega}$, (entropies are with respect to the Koashi-Imoto decomposition). Below, we show that we can obtain this result as a corollary of Theorem~\ref{thm: E-assisted}.\\
\begin{proposition}
 The optimal entanglement-assisted rate for the simulation of the identity channel
    $\cid:A \to A$ is equal to $a(\rho^{AR},0)=S(CQ)_{\omega}-\frac{1}{2}S(C)_{\omega}$. 
\end{proposition}
\begin{proof}
For $\cN=\cid$, the function at $\gamma=0$ is 
    \begin{align}
        a(\rho^{AR}, 0)&:= \min_{\Lambda:A\to A} \frac{1}{2} I(A:RR')_{\tau}
        \quad \quad \text{s.t.} \quad F(\rho^{AR},\tau^{AR}) =1. \nonumber
    \end{align}
Consider the KI-decomposition of $\rho^{AR}$ only with systems $CQ$ and its purification 
\begin{align}
\omega^{CNQR}&=\sum_c p_c \proj{c}^C \ox \omega_c^{QR} \nonumber\\
\ket{\omega}^{CQRR'C'}&=\sum_c \sqrt{p_c} \ket{c}^C \ox \ket{\omega_c}^{QRR'}\ox \ket{c}^{C'},
\end{align}
where  $R'C'$ are purifying systems.
Note that $a(\omega^{CQR}, 0)=a(\rho^{AR}, 0)$ holds since there are CPTP maps in both directions $\cT:A \to CQ$ and $\cR:CQ \to A$, and applying CPTP maps only increases the fidelity and lowers the mutual information. So, we evaluate the function below
\begin{align}
        a(\omega^{CQR}, 0)&:= \min_{\Lambda:CQ\to CQ} \frac{1}{2} I(CQ:RR')_{\tau}
        \quad \quad \text{s.t.} \quad F(\omega^{CQR},\tau^{CQR}) =1. \nonumber
    \end{align}
Let $\Lambda_0:C \to CC''$ be a map which copies system $C$ to another register $C''$. This gives the state
\begin{align}
\ket{\tau}^{CQRR'C'C''}&=\sum_c \sqrt{p_c} \ket{c}^C \ox \ket{\omega_c}^{QRR'}\ox \ket{c}^{C'}\ox \ket{c}^{C''},
\end{align}
and the mutual information evaluates to
\begin{align}
I(CQ:RR'C')_{\tau}&=S(CQ)_{\tau}+S(RR'C')_{\tau}-S(CQRR'C')_{\tau} \nonumber\\
&=S(CQ)_{\tau}+S(CC''Q)_{\tau}-S(C'')_{\tau} \nonumber\\
&=S(CQ)_{\omega}+S(CQ)_{\omega}-S(C)_{\omega}, \nonumber
\end{align}
where the second follows since the overall state on $CQRR'C'C''$ is pure. The last equality follows because $C''$ is a copy of $C$.
This implies that $2a(\omega^{CNQR}, 0)\geq 2S(CQ)_{\omega}-S(C)_{\omega}$.
Now, we show that $\Lambda_0:C \to C''$ is optimal. By KI-theorem the isometric extension $U_{\Lambda}$ of any CPTP map $\Lambda:CQ \to CQ$ that preserves $\omega^{CQR}$ can only act as follows
\begin{align}
\ket{\nu}^{CQRR'C'M}&=(U_{\Lambda}\ox \1_{RR'})\ket{\omega}^{CQRR'} \nonumber\\
&=\sum_c \sqrt{p_c} \ket{c}^C \ox \ket{\omega_c}^{QRR'}\ox \ket{c}^{C'}\ox \ket{v_c}^{M}.
\end{align}
Hence, the mutual information is bounded as
\begin{align}
I(CQ:RR'C')_{\nu}&=S(CQ)_{\nu}+S(RR'C')_{\nu}-S(CQRR'C')_{\nu} \nonumber\\
&=S(CQ)_{\nu}+S(CMQ)_{\nu}-S(M)_{\nu} \nonumber\\
&=S(CQ)_{\omega}+S(CQ)_{\omega}+S(M|C)_{\omega}-S(M)_{\nu}, \nonumber\\
&=S(CQ)_{\omega}+S(CQ)_{\omega}-S(M)_{\nu}, \nonumber
\end{align}
in the third line the $S(M|C)_{\omega}=0$ because the state on $M$ given $C=c$ is pure. In the last line  $S(M)_{\nu} $ is maximized if $\ket{v_c}^{M}$ are orthogonal, that is $S(M)_{\nu} =S(C)_{\nu} $.

\end{proof}

\section{Unassisted Simulation}\label{sec: unassisted}
In this section, we obtain achievability and converse bounds for the unassisted qubit rate. 
The unassisted model refers to the case where the encoder and  decoder do not share or distill any entanglement, namely registers $A_0,A_1,B_0$ and $B_1$ are trivial registers. 
\begin{theorem}\label{thm: unassisted}
    For the unassisted  simulation of the channel $\cN:A \to BK$, the following  regularized rate 
    is achievable where $u(\rho,\gamma)$ is defined in Definition~\ref{def: a(rho,epsilon)}
\begin{align}
    Q^* \leq  \lim_{m \to \infty}  \frac{1}{m} u(\rho^{\ox m}, \frac{1}{m^9}). \nonumber
\end{align}
 %
Moreover, any achievable quantum rate is lower bounded as
    \begin{align}
    Q^* \geq \lim_{\gamma \to 0} \lim_{m \to \infty} \frac{1}{m}u(\rho^{\ox m},\gamma). \nonumber
\end{align}
\end{theorem}


\begin{proof}
The proof of the direct part is as follows.
Consider $m$ copies of the state $\rho^{AR}$ as a single state and  
the optimal CPTP maps in Definition~\ref{def: a(rho,epsilon)}, i.e. $\Lambda_2:A^m\to B^m K^m$ and $\Lambda_3:E\to E'$
at $\gamma_m$. 
Let $U_{\Lambda_2}:A^m\to B^m K^m E$ and $U_{\Lambda_3}:E\to E' E''$ be the corresponding Stinespring dilation isometries of $\Lambda_2$ and $\Lambda_3$, respectively.
 Alice applies $U_{\Lambda_2}$ and $U_{\Lambda_3}$ to $m$ copy of the
state $\rho^{AR}$ as follows
\begin{align}
    \ket{\tau_2}^{B^m K^m ER^m{R'}^m}&=(U_{\Lambda_2} \ox \1_{R^m{R'}^m})\ket{\rho}^{A^m R^m{R'}^m }, \nonumber\\
    \ket{\tau_3}^{B^m K^m E'E''R^m{R'}^m}&=(U_{\Lambda_3} \ox \1_{B^mK^mR^m{R'}^m})\ket{\tau_2}^{B^m K^m ER^m{R'}^m}, \nonumber 
\end{align}
where $\ket{\rho}^{A^m R^m{R'}^m }$ is a purification of $(\rho^{AR})^{\ox m}$, and systems
$R^m$ and ${R'}^m$ are held by inaccesible reference systems. 
%
%
Then Alice and Bob perform Schumacher compression on $k$ copies of $\tau_3^{ B^m K^m E'E''R^m{R'}^m}$ to send systems $B^mE'$ from Alice to Bob,
assuming that the systems $K^mE''R^m{R'}^m$ are held by a reference. 
The rate of this 
protocol is equal to $S(B^m E')_{\tau_3} +\eta_k$. 
Moreover, this asymptotic protocol preserves the fidelity with $k$ copies of $
\tau_3^{ B^m K^m E'E''R^m{R'}^m}$, i.e.
\begin{align}
F((\tau_3^{ B^m K^m E'E''R^m{R'}^m})^{\ox k},\upsilon^{ B^{mk} K^{mk}  E'^{k} E''^{k} R^{mk} {R'}^{mk}})\geq 1-\epsilon_k,
\end{align}
where $(\upsilon^{ B^m K^m E'E''R^m{R'}^m})^{\ox k}$ is the decoded state of Schumacher compression, and $\eta_k$ and $\epsilon_k$ converge to 0 as $k$ converges to $\infty$. 
More precisely, in Schumacher compression we can choose $\epsilon_k = (\frac{\log|B^mE'|}{\cso(\sqrt{k})})^2$. From Remark~\ref{remark: min vs inf}, we obtain the bound  $|E'| \leq ( |A|\cdot |B|\cdot |K|)^{2m}$,  hence, the error can be bounded as
\begin{align}
\epsilon_k \leq (\frac{2m\log  |A|\cdot |B|^{\frac{3}{2}}\cdot |K|}{\cso(\sqrt{k})})^2.
\end{align}
At the last step of the proof, we will take the limit $m \to \infty$. Therefore,  to have vanishing error in the limit of $k \to \infty$, we may choose $m=k^{\frac{1}{4}}$

Tracing out systems $E'^kE''^kR'^{mk}$ only increases the fidelity, hence, we obtain
\begin{align}
F((\tau_3^{ B^m K^m R^m})^{\ox k},\upsilon^{ B^{mk} K^{mk}   R^{mk} })\geq 1-\epsilon_k. 
\end{align}
Hence, Fuchs-van de Graaf inequality implies that
\begin{align}\label{eq: D1}
\frac{1}{2}\norm{(\tau_3^{ B^m K^m R^m})^{\ox k}-\upsilon^{ B^{mk} K^{mk} R^{mk} }}_1\leq\sqrt{ 1-(1-\epsilon_k)^2}. 
\end{align}
In what follows, we show that 
\begin{align}\label{eq: D2}
\frac{1}{2}\norm{(\sigma^{ BKR})^{\ox mk}-\upsilon^{ B^{mk} K^{mk} R^{mk} }}_1\leq\sqrt{ 1-(1-\epsilon_k)^2}+k\sqrt{1-(1-\gamma_m)^2} \leq \sqrt{2\epsilon_k}+k\sqrt{2\gamma_m}. 
\end{align}
By definition of $\tau_3$ and applying  Fuchs-van de Graaf inequality we have \begin{align}
\frac{1}{2}\norm{\tau_3^{ B^m K^m R^m}-(\sigma^{ BKR})^{\ox m}}_1 \leq \sqrt{1-(1-\gamma_m)^2}.
\end{align}
We prove that the trace distance between $k$-fold tensor power of the above states   
is bounded by $\sqrt{1-(1-\gamma_m)^2}$.
To this end, we apply triangle inequality as follows
\begin{align}
\frac{1}{2}&\norm{(\tau_3^{ B^m K^m R^m})^{\ox k}-(\tau_3^{ B^m K^m R^m})^{\ox k-1}\ox (\sigma^{ BKR})^{\ox m} +(\tau_3^{ B^m K^m R^m})^{\ox k-1}\ox (\sigma^{ BKR})^{\ox m} -(\sigma^{ BKR})^{\ox mk}}_1 \nonumber\\
& \leq \frac{1}{2}\norm{(\tau_3^{ B^m K^m R^m})^{\ox k}-(\tau_3^{ B^m K^m R^m})^{\ox k-1}\ox (\sigma^{ BKR})^{\ox m}}_1 +\frac{1}{2}\norm{(\tau_3^{ B^m K^m R^m})^{\ox k-1}\ox (\sigma^{ BKR})^{\ox m} -(\sigma^{ BKR})^{\ox mk}}_1 \nonumber\\
&\leq \sqrt{1-(1-\gamma_m)^2}+ \frac{1}{2}\norm{(\tau_3^{ B^m K^m R^m})^{\ox k-1}\ox (\sigma^{ BKR})^{\ox m} -(\sigma^{ BKR})^{\ox mk}}_1 \nonumber \\
&= \sqrt{1-(1-\gamma_m)^2}+ \frac{1}{2}\norm{(\tau_3^{ B^m K^m R^m})^{\ox k-1} -(\sigma^{ BKR})^{\ox m(k-1)}}_1,
\end{align}
We apply the above procedure for $\frac{1}{2}\norm{(\tau_3^{ B^m K^m R^m})^{\ox k-1} -(\sigma^{ BKR})^{\ox m(k-1)}}_1$, and repeat this $k-1$ times and obtain 
\begin{align}
\frac{1}{2} \norm{(\tau_3^{ B^m K^m R^m})^{\ox k}-(\sigma^{ BKR})^{\ox mk} }_1 \leq k\sqrt{1-(1-\gamma_m)^2}. \nonumber
\end{align}
From the above inequality and Eq.~(\ref{eq: D1}) we obtain the desired bound in Eq.~(\ref{eq: D2}).
In this equation $\epsilon_k$ converges to 0 as $k$ grows. 
We already set $m=k^{\frac{1}{4}}$, hence, by letting  $\gamma_m=\frac{1}{k^{2.25}}$ 
the upper bound in Eq.~(\ref{eq: D2}) converges to 0 as $k$ converges to $\infty$.
The rate of the above protocol is $S(B^m E')_{\tau_3} +\eta_k$. Dividing this by $m$ we obtain the rate 
\begin{align}
 \frac{1}{m} \left(S(B^m E')_{\tau_3} +\eta_k\right)&=\frac{1}{m} u(\rho^{\ox m}, \frac{1}{k^{2.25}})+\frac{\eta_k}{m}.
\end{align}
Thus, the asymptotic rate of the above protocol is 
\begin{align}
\lim_{k \to \infty} \frac{1}{m} \left(S(B^m E')_{\tau_3} +\eta_k\right)&= \lim_{k \to \infty} \frac{1}{m} u(\rho^{\ox m}, \frac{1}{k^{2.25}})+\lim_{k \to \infty} \frac{\eta_k}{m} \nonumber\\
&=\lim_{m \to \infty} \frac{1}{m} u(\rho^{\ox m}, \frac{1}{m^9}).  
\end{align}

\medskip

For the converse bound of the unassisted case, Eq.~(\ref{eq: unassisted converse}) is reduced to
\begin{align}
    n\mathcal{Q}(n,\epsilon )&\geq S(M)_{\nu} \nonumber \\
    &= S(B^n W_B)_{\xi} \nonumber \\
    & \geq u(\rho^{\ox n}, \epsilon), \nonumber
\end{align}
where the second line is due to the decoding isometry.
The last line follows from Definition~\ref{def: a(rho,epsilon)}. 
We remind that the optimal qubit rate is defined as $  \lim_{\epsilon \to 0} \limsup_{n \to \infty} \inf_{\mathcal{C}_n,\mathcal{D}_n} \mathcal{Q}(n,\epsilon )$. So, the converse follows from dividing both sides by $n$ and taking the limit of $n \to \infty$ and $\epsilon \to 0$.
%
%
\end{proof}





In general, it is not obvious if the converse and achievable rates of the above theorem are equal. Below, we provide examples for which the two rate are equal. 


\begin{proposition}\label{prp: pure  input state}
If we assume the input state $\rho^{AR}$ is pure, as in \cite{Bennett2014a}, then   the  optimal unassisted simulation rate is $\lim_{n \to \infty}  \lim_{\epsilon \to 0}  \frac{1}{n}u(\rho^{\ox n},\epsilon)= \lim_{\epsilon \to 0} \lim_{n \to \infty} \frac{1}{n}u(\rho^{\ox n},\epsilon)=\lim_{n \to \infty} \frac{1}{n}u(\rho^{\ox n},0)=E_p^{\infty}(B:KR)_{\sigma}$.
%
\end{proposition}

\begin{proof}
We remind that for a pure state $\ket{\rho}^{AR}$
 \begin{align}
        u((\proj{\rho}^{AR})^{\ox n}, \epsilon)&:= \min_{\substack{\>\>\>\Lambda_2:A \to BK\\ \Lambda_3:E  \to E'}} S(B^nE')_{\tau_3}
        \quad \quad \quad \>\> \> \text{s.t.} \quad F((\sigma^{BKR})^{\ox n},\tau_3^{B^nK^nR^n}) \geq 1-\epsilon,\nonumber
    \end{align}
where $\ket{\tau_3}^{B^nK^nE'E''R^n}=(U_{\Lambda_3} \ox \cid_{B^nK^nR^n})(U_{\Lambda_2} \ox \cid_{R^n})\ket{\rho}^{A^nR^n}$. Here, $U_{\Lambda_2}$ and $U_{\Lambda_3}$ are isometric extensions of ${\Lambda_2}$ and ${\Lambda_3}$, respectively.  
By definition $u((\proj{\rho}^{AR})^{\ox n}, 0) \geq u((\proj{\rho}^{AR})^{\ox n}, \epsilon)$ holds for $\epsilon \geq 0$.
For the other direction, we obtain
\begin{align}\label{eq: EoP lower bound}
     S(B^n E')_{\tau_3}  &\geq  S(B^n E')_{\sigma} +  \sqrt{2\epsilon} \log(|B^nE'|) +  h(\sqrt{2\epsilon}) \nonumber \\ 
     &\geq  u((\proj{\rho}^{AR})^{\ox n}, 0)+   \sqrt{2\epsilon} \log(|B^nE'|) +   h(\sqrt{2\epsilon}) \nonumber\\
    &=  E_p(B^n:K^nR^n)_{\sigma} +  \sqrt{2\epsilon} \log(|B^nE'|) +   h(\sqrt{2\epsilon})  
\end{align}
where in the first line the entropy is with respect to the state $\sigma^{B^nE'}=(\cM^{G^n \hookrightarrow E'}  \ox\cid_{B^nK^nR^n})\proj{\sigma}^{B^nK^n G^n R^n}$, that is $E'$ is obtained by applying a CPTP map acting on the environment system $G^n$ of the channel $\cN^{\ox n}$.
This line is due to Uhlmann's theorem \cite{UHLMANN1976}: 
the  state $\tau_3^{B^nK^nR^n}$ has 
$1-\epsilon$ fidelity with $\sigma^{B^nK^nR^n}$, hence, there is a purification  $V^{G^n \hookrightarrow E'E''}\ket{\sigma}^{B^nK^nG^nR^n}$ of the state $\sigma^{B^nK^nR^n}$, which has $1-\epsilon$ fidelity with  the purified state $\tau_3^{B^nK^nR^n E'E''}$. By tracing out system $E''$ the fidelity only increases. hence, the first line holds 
because $\tau_3$ is $2\sqrt{2\epsilon}$-close 
to $\sigma$ in trace norm; then the inequality follows 
by applying the Fannes-Audenaert inequality.
Finally the proposition  follows 
by dividing the above inequality by $n$ and taking the limit of $n\to \infty$ and $\epsilon \to 0$.
The $\epsilon$-terms  vanish because the dimension of system $E'$ is bounded as explained in Remark~\ref{remark: min vs inf}.
\end{proof}
\begin{remark}
Indeed for a general mixed input state $\rho^{AR}$, the optimal unassisted rate is lower bounded by $E_p^{\infty}(B:KR)_{\sigma}$. 
This lower bound can be obtained similarly to the lower bound we derive in Eq.~(\ref{eq: EoP lower bound}). 
However, this rate cannot be achievable in the general case since this entanglement of purification is the minimal possible entropy of systems $B^nE'$, where $E'$ is obtained by applying a CPTP map on system $E$, as defined in Definition~\ref{def: a(rho,epsilon)}, and system $R'$ which purifies the source $\rho^{AR}$, and it is inaccessible to the encoder.
%
\end{remark}

\begin{proposition}
 The optimal unassisted rate for the simulation of the identity channel
    $\cid:A \to A$ is equal to $ \lim_{\epsilon \to 0} \lim_{n \to \infty} \frac{1}{n}u(\rho^{\ox n},\epsilon)=\lim_{n \to \infty} \frac{1}{n}u(\rho^{\ox n},0)=S(CQ)_{\omega}$. 
\end{proposition}

\begin{proof}
As explained in Sec.~\ref{sec: Miscellaneous} below Theorem~\ref{thm: KI decomposition}, there are unitary  CPTP maps in both directions $U_{\KI}:A \to CNQ$ and $\cR:CNQ \to A$ which relate a state $\rho^{AR}$ to its KI-decomposition 
$\omega^{CNQR}$. This implies that 
$u((\rho^{AR})^{\ox n}, \epsilon)=u((\omega^{CNQR})^{\ox n}, \epsilon)$  since applying unitary CPTP maps do not change the entropy and only increases the fidelity.
Hence for  $\cN=\cid$ 
 \begin{align}
        u(({\rho}^{AR})^{\ox n}, \epsilon)&=u(({\omega}^{CNQR})^{\ox n}, \epsilon) \nonumber\\
        &= \min_{\substack{\>\>\> \>\>\> \>\>\>\Lambda_2:CNQ \to CNQ\\ \>\!\!\!\!\!\Lambda_3:E  \to E'}} S(C^nN^nQ^nE')_{\tau_3}
        \quad \quad \quad \>\> \> \text{s.t.} \quad F((\omega^{CNQR})^{\ox n},\tau_3^{C^n N^n Q^nR^n}) \geq 1-\epsilon,\nonumber
    \end{align}
where $\ket{\tau_3}^{C^n N^n Q^nE'E''R^nR'^n}=(U_{\Lambda_3} \ox \cid_{C^n N^n Q^nR^nR'^n})(U_{\Lambda_2} \ox \cid_{R^n}) (\ket{\omega}^{\otimes n})^{C^n N^n Q^nR^nR'^n}$. Here, $U_{\Lambda_2}$ and $U_{\Lambda_3}$ are isometric extensions of ${\Lambda_2}$ and ${\Lambda_3}$, respectively.  
We obtain
\begin{align}
S(C^n N^n Q^nE')_{\tau_3} 
&=S(C^n Q^n)_{\tau_3} +S(E'N^n|C^n  Q^n)_{\tau_3}  \nonumber\\
&\geq nS(CQ)_{\omega} +S(E'N^n|C^n  Q^n)_{\tau_3}
+   \sqrt{2\epsilon} n\log(|C|\cdot|Q|) +   h(\sqrt{2\epsilon}) \nonumber\\
&= nS(CQ)_{\omega} -I(E'N^n:Q^n|C^n  )_{\tau_3} +S(E'N^n)_{\tau_3}
+   \sqrt{2\epsilon} n\log(|C|\cdot|Q|) +   h(\sqrt{2\epsilon}) \nonumber\\
&\geq nS(CQ)_{\omega} -nJ_{\epsilon}(\omega) +S(E'N^n)_{\tau_3}
+   \sqrt{2\epsilon} n\log(|C|\cdot|Q|) +   h(\sqrt{2\epsilon}) \nonumber\\
&\geq nS(CQ)_{\omega} -nJ_{\epsilon}(\omega) 
+   \sqrt{2\epsilon} n\log(|C|\cdot|Q|) +   h(\sqrt{2\epsilon}) \nonumber\\
\end{align}
where the second line follows
because $\tau_3$ is $2\sqrt{2\epsilon}$-close 
to $\omega$ in trace distance; then the inequality follows 
by applying the Fannes-Audenaert inequality. 
The third line is to the definition of the conditional mutual information.
In the penultimate line $J_{\epsilon}(\omega)\to 0$ as $\epsilon \to 0$ as defined and proven in \cite{general_mixed_state_compression,ZK_mixed_state_ISIT_2020,ZBK_PhD}
.
By dividing by $n$ and taking the limits, we obtain 
\begin{align}
 \lim_{\epsilon \to 0} \lim_{n \to \infty} \frac{1}{n}u(\rho^{\ox n},\epsilon)\geq S(CQ)_{\omega}.
\end{align}
Also, we  show $u({\omega}^{CNQR} , 0) = S(CQ)_{\omega}$ as follows
 \begin{align}
      u({\omega}^{CNQR} , 0) 
        &= \min_{\substack{\>\>\> \>\>\> \>\>\>\Lambda_2:CNQ \to CNQ\\ \>\!\!\!\!\!\Lambda_3:E  \to E'}} S(C^nN^nQ^nE')_{\omega} \nonumber\\        
 &= S(CQ)_{\omega}   +S(NE'|C) \nonumber\\
  &\geq  S(CQ)_{\omega},    \nonumber
\end{align}
where the second line is by Theorem~\ref{thm: KI decomposition}, namely, the environment system of a CPTP map, which preserves the state, acts only on the redundant system. The inequality above can be saturated by choosing  $\Lambda_2$ as a map which traces out system $N$ and for given $c$ outputs a pure state $\proj{\omega_c}^{NE'}$
and by letting $\Lambda_3=\cid$.
Finally, by definition
\begin{align}
 u({\omega}^{CNQR} , 0) \geq \lim_{n \to \infty} \frac{1}{n}u(({\omega}^{CNQR})^{\ox n} ,0) \geq \lim_{\epsilon \to 0} \lim_{n \to \infty} \frac{1}{n}u(({\omega}^{CNQR})^{\ox n} ,  \epsilon),
\end{align}
and this completes the proof. 
\end{proof}

\begin{proposition}
For a fully classical input state  $\rho^{AR}=\sum_x p_x \proj{x}^A\ox \proj{x}^R$  the  optimal unassisted simulation rate is $\lim_{n \to \infty} \frac{1}{n}u(\rho^{\ox n},0)= \lim_{\epsilon \to 0} \lim_{n \to \infty} \frac{1}{n}u(\rho^{\ox n},\epsilon)=E_p^{\infty}(B:KR)_{\sigma}$.
\end{proposition}

\begin{proof}
The proof is similar to the proof of Proposition~\ref{prp: pure  input state} as follows
%
%
\begin{align}\label{eq: EoP lower bound 2}
     S(B^n E')_{\tau_3}  &\geq  S(B^n E')_{\sigma} +  \sqrt{2\epsilon} \log(|B^nE'|) +  h(\sqrt{2\epsilon}) \nonumber \\ 
     &\geq  u((\proj{\rho}^{AR})^{\ox n}, 0)+   \sqrt{2\epsilon} \log(|B^nE'|) +   h(\sqrt{2\epsilon}) \nonumber\\
    &=  E_p(B^n:K^nR^n)_{\sigma} +  \sqrt{2\epsilon} \log(|B^nE'|) +   h(\sqrt{2\epsilon})  
\end{align}
where in the first line the entropy is with respect to the state ${\sigma}^{B^nE'K^nR^nR'^n}=(\cM^{G^n R'^n \hookrightarrow E'}  \ox\cid_{B^nK^nR^n})\proj{\sigma}^{B^nK^n G^n R^n R'^n}$, and the inequality
is obtained by Uhlmann's theorem and Fannes-Audenaert inequality. 
The difference with Proposition~\ref{prp: pure  input state} is that the map $\cM^{G^n R'^n \hookrightarrow E'}$  acts on $R'^n$ as well. To obtain the second inequality note that $R'^n$ in the purified input state $\sum_{x^n} \sqrt{p_{x^n}}\ket{x^n}^{A^n} \ket{x^n}^{R^n} \ket{x^n}^{R'^n}$ is another copy of system $A^n$, hence, the desired  $E'$ system in Uhlmann's construction can be obtained by applying a map only on $A^n$. 
Finally the proposition  follows 
by dividing the above inequality by $n$ and taking the limit of $n\to \infty$ and $\epsilon \to 0$.
The $\epsilon$-terms  vanish because the  dimension of system $E'$ is bounded as explained in Remark~\ref{remark: min vs inf}.
\end{proof}




\section{Discussion}\label{sec: discussion}

We consider an asymptotic i.i.d. simulation of an arbitrary channel $\cN:A \to BK$, with an isometric extension $U_{\cN}: A \to BKG$,
acting as $(\cN \ox \cid^R)\rho^{AR}$    
on a general mixed input state shared with a reference system $R$.
An encoder, Alice, has access to $A^n$, and the goal is to simulate system $B^n$ at the decoder side, Bob, and system $K^n$ at her side.
This general definition captures various considerations of \cite{Bennett2014a} in a single model: classical, quantum, feedback, non-feedback.
The two extreme cases of the fully classical and fully quantum models in \cite{Bennett2014a} are realized by constraining 
$\rho^{AR}$ to be either a classical state or a pure quantum state, respectively.
The non-feedback and feedback models  are realized by  constraining 
$K=\emptyset$ or $K=G$ (the environment of an isometric extension), respectively.
We also recover the general mixed-state compression of \cite{general_mixed_state_compression,ZK_mixed_state_ISIT_2020} and the visible compression of mixed-state ensembles considered in \cite{Hayashi2006} by constraining $\cN$ to be an identity channel and the input $\rho^{AR}$ to be fully classical, respectively

We define two functionals in Definition~\ref{def: a(rho,epsilon)} and show that they characterize the simulation rates for all models, irrespective of constraints.
We prove that the optimal entanglement-assisted rate for the simulation of the channel     $\cN:A \to BK$ is equal to $a(\rho^{AR},0)$. This is a quantity  depends only on a single copy of the input state  $\rho^{AR}$ and channel $\cN$.
 For the unassisted  simulation, we prove a rate $\lim_{m \to \infty}  \frac{1}{m} u(\rho^{\ox m}, \frac{1}{m^9})$ is achievable, and the lower bound $ \lim_{\gamma \to 0} \lim_{m \to \infty} \frac{1}{m}u(\rho^{\ox m},\gamma)$ holds.  Even though we could not prove these two rates are equal in general, we provide multiple examples that these two bound match and can be simplified to single-letter  quantities. 

Several directions remain open for further exploration. One immediate avenue is to investigate whether the unassisted rate functional $u(\rho, \gamma)$ admits a tighter or more computable characterization, potentially leading to a deeper understanding of the regularization behavior.  Another promising direction is to consider local error instead of global error criterion or more generally a rate-distortion model.
Finally, an open question left by our work is determining the resource trade-off between shared randomness, shared entanglement, and quantum communication. In our model, we only consider shared entanglement and quantum communication; including shared randomness as an additional resource completes the picture from a resource-theoretic perspective.

\bigskip
\noindent \textbf{Acknowledgments.}
ZBK was supported by the Marie Sk{\l}odowska-Curie Actions (MSCA) European Postdoctoral Fellowships (Project 101068785-QUARC) 
and the Ada Lovelace Postdoctoral Fellowship at Perimeter Institute for Theoretical Physics.
DL received support from NSERC discovery grant and NSERC Alliance grant under the project QUORUM.
ZBK and DL, ia the Perimeter Institute, are supported in part by the
Government of Canada through the Department of Innovation, Science and Economic Development
and by the Province of Ontario through the Ministry of Colleges and Universities.


\appendix

\section{Proof of Lemma~\ref{lemma: decoupling}}
\label{sec: Proof of decoupling lemma}

\begin{proof}
From the fidelity criterion we obtain
\begin{align}\label{eq: operator norm xi A_1B1}
    1-\epsilon &\leq F\left(\sigma^{B^nK^nR^n}\otimes \proj{\Phi}^{A_1B_1} , \xi_n^{B^nK^nR^n A_1 B_1}\right) \nonumber \\
    & \leq F\left(\proj{\Phi}^{A_1B_1}, {\xi_{n}}^{A_1 B_1} \right) \nonumber \\
    &=\sqrt{\bra{\Phi} \xi_n^{A_1 B_1} \ket{\Phi}} \nonumber \\
    &\leq \sqrt{ \norm{\xi_n^{A_1 B_1}}_{\infty}},
\end{align}
where the second line is due to monotonicity of the fidelity under partial trace.
The last line follows from the definition of the operator norm.
Now, consider the Schmidt decomposition of the state 
$\ket{\xi_{n}}^{B^n K^n W_A W_B R^n {R'}^n A_1 B_1}$ with respect to the partition
$B^n K^n W_A W_B R^n {R'}^n: A_1 B_1 $, i.e.
\begin{align*}
\ket{\xi_{n}}^{B^n K^n W_A W_B R^n {R'}^n A_1 B_1}
     = \sum_{i} \sqrt{\lambda_i}\ket{v_i}^{B^n K^n W_A W_B R^n {R'}^n} \ket{w_i}^{A_1 B_1}. 
\end{align*}
Considering the above decomposition, we obtain
\begin{align}
  \label{eq:almost-pure}
   F&\left( \ketbra{\xi_{n}}^{B^n K^n W_A W_B R^n {R'}^n A_1 B_1},
                             {\xi_{n}}^{B^n K^n W_A W_B R^n {R'}^n } \ox {\xi_{n}}^{ A_1 B_1} \right) \nonumber\\
    &= \sqrt{\bra{\xi_{n}}{\xi_{n}}^{B^n K^n W_A W_B R^n {R'}^n } \ox {\xi_{n}}^{ A_1 B_1} 
                               \ket{\xi_{n}}}                                                                 \nonumber\\
    &= \sum_i \lambda_i^{\frac32} \nonumber\\
    & \geq \norm{\xi_n^{A_1 B_1}}_{\infty}^{\frac32} \nonumber\\
    &\geq (1-\epsilon)^3 
     \geq 1 - 3\epsilon,   
\end{align}
where the last line follows from Eq.~(\ref{eq: operator norm xi A_1B1}).
Finally, 
By the Alicki-Fannes inequality (Lemma \ref{AFW lemma}), this implies
\begin{align}
  \label{decoupling_I}
  I(B^n K^n W_A W_B R^n {R'}^n :A_1 B_1)_\xi
     &=    S(A_1 B_1)_\xi-S(A_1 B_1|B^n K^n W_A W_B R^n {R'}^n)_\xi \nonumber \\ 
     &\leq 2\sqrt{6\epsilon} \log(|A_1| |B_1|) + 2 h(\sqrt{6\epsilon}) \nonumber \\
     &= 4n\sqrt{6\epsilon} \log(d_1) + 2 h(\sqrt{6\epsilon}) 
           =: n \delta(n,\epsilon),  
\end{align}
where we assume that $|A_1|={d_1}^n$ for some $d_1>0$.
\end{proof}

\section{Proof of Lemma~\ref{lemma: f properties}}
\label{sec: Proof of Lemma f properties}

\begin{enumerate}
\item The definition of the function 
  directly implies that it is a non-decreasing function of $\epsilon$.

\item Let $U_1:A \hookrightarrow BKZ$ and 
  $U_2:A \hookrightarrow BKZ$ be the isometric extensions of the maps attaining the 
  minimum for $\gamma_1$ and $\gamma_2$, respectively, which act as 
  follows on the purified state $\ket{\rho}^{AR R'}$ 
 \begin{align*}
     \ket{\tau_1}^{BKZR R'}
        &=(U_1 \otimes \1_{R  R'}) \ket{\rho}^{ARR'}
    \text{ and } \\
    \ket{\tau_2}^{BKZRR'}
        &=(U_2 \otimes \1_{R  R'}) \ket{\rho}^{ARR'}. 
 \end{align*}
  For $0\leq \lambda \leq 1$, define the isometry 
  $U_0:A \hookrightarrow BKZ F F'$ which acts as 
  \begin{equation}
    \label{eq: isometry U in convexity}
    U_0 := \sqrt{\lambda} U_1 \otimes \ket{00}^{FF'} + \sqrt{1-\lambda} U_2 \otimes \ket{11}^{FF'},
  \end{equation}
  where systems $F$ and $F'$ are qubits, 
  and by applying $U_0$ we obtain
 \begin{align}
     (U_0  \otimes \1_{R  R'}) \ket{\rho}^{AR R'} 
      =\sqrt{\!\lambda}\ket{\tau_1}^{BKZR R'} \!\ket{00}^{\!FF'}
        + \sqrt{1-\lambda}\ket{\tau_2}^{BKZR R'\!} \ket{11}^{\!\!F\!F'},
 \end{align}
  and the reduced state on the systems 
  $BKRR'F$ is 
  \begin{align} \label{eq: tau in convexity proof}
    \tau^{BKRR'F} 
      =\lambda \tau_1^{BKRR'} \ox \proj{0}^F+ (1-\lambda) \tau_2^{BKRR'} \ox \proj{1}^F. 
  \end{align}  
  The fidelity for the state $\tau^{BKR}$ is bounded as follows:
  \begin{align}\label{eq:fidelity in convexity}
    F&(\sigma^{BK R} ,\tau^{BK R} )\nonumber \\
      &= F(\sigma^{BK R} ,\lambda \tau_1^{BK R}
        + (1-\lambda) \tau_2^{BKR}) \nonumber \\
      &\!= \!F(\!\lambda \sigma^{BKR}+\!\!(1\!\!-\!\!\lambda)\sigma^{BKR},
           \lambda  \tau_1^{BKR}
            +(1-\lambda) \!\tau_2^{BK R}) \nonumber\\
      &\geq \lambda F( \sigma^{BKR},\tau_1^{BK R})
            +(1-\lambda)F\!(  \sigma^{BK R},\tau_2^{BK R}\!) \nonumber\\
     &\geq 1-\left( \lambda\gamma_1 +(1-\lambda)\gamma_2 \right).
  \end{align}
  The first inequality is due to simultaneous concavity of the fidelity in both
  arguments;
  the last line follows by the definition of the isometries $U_1$ and $U_2$.
  Thus, the isometry $U_0$ yields a fidelity of at least 
  $1-\left( \lambda\gamma_1 +(1-\lambda)\gamma_2 \right) =: 1-\gamma$.
  Let $Z_0=Z FF'$ denote the environment of the isometry $U_0$ defined above. 
  We can obtain
  \begin{align}
   2 a(\rho,\epsilon)
                     &\leq I(B:RR')_{\tau} \nonumber\\ 
                     &\leq I(BF:RR')_{\tau} \nonumber\\ 
                     &= I(F:RR')_{\tau}+I(B:RR'|F)_{\tau} \nonumber\\ 
                     &= I(B:RR'|F)_{\tau} \nonumber\\ 
                     &=\lambda I(B:RR')_{\tau_1}+(1-\lambda)I(B:RR')_{\tau_2} \nonumber\\
                     &= \lambda a(\rho,\gamma_1)+(1-\lambda)a(\rho,\gamma_2),
  \end{align}
  where the quantum mutual information is with respect to the  state $\tau$ in Eq.~(\ref{eq: tau in convexity proof}).
The second line is due to the data processing inequality.  The fourth line holds because systems $RR'$ are independent from $F$.

\item We  prove 
  $a(\rho_1^{A_1R_1} \ox \rho_2^{A_2R_2}, \gamma) \geq  a(\rho_1^{A_1R_1}, \gamma)+a(\rho_2^{A_2R_2}, \gamma)$.
    \begin{align}
        a(\rho_1^{A_1R_1} \ox \rho_2^{A_2R_2}, \epsilon)&:= \min_{\Lambda:A_1 A_2\to B_1K_1 B_2K_2} \frac{1}{2} I(B_1B_2:R_1R'_1R_2R'_2)_{\tau} 
        \quad \text{s.t.} \nonumber\\
         &\quad \quad\quad\quad\quad\quad\quad\quad\quad\quad\quad F(\sigma_1^{B_1K_1R_1}\ox \sigma_2^{B_2K_2R_2},\tau^{B_1K_1R_1 B_2K_2R_2}) \geq 1-\gamma,\nonumber
    \end{align}
   where the quantum mutual information is with respect to the state
   \begin{align}\label{eq:U0-action}
       \ket{\tau}^{B_1 B_2 K_1 K_2 Z R_1R'_1R_2 R'_2}=(U_0 \ox \1_{R_1R'_1R_2 R'_2})(\ket{\rho_1}^{A_1  R_1R'_1} \ox \ket{\rho_2}^{A_2  R_2R'_2})
   \end{align}
  and the isometry 
  $U_0:A_1 A_2 \hookrightarrow B_1 B_2 K_1 K_2 Z$
  is the Stinespring dilation of the map attaining the minimum, and $Z$ is the environment system. The isometry  acts on the  purified source states with purifying 
  systems $R'_1$ and $R'_2$.

  We can define an isometry 
  $U_1:A_1 \hookrightarrow B_1 K_1 Z_1$ 
  acting only on system $A_1$, by letting
  $U_1 = (U_0 \otimes \1_{R_1R'_1R_2 R'_2})(\1_{R_1R'_1} \otimes \ket{\rho_2}^{A_2  R_2R'_2})$
  and with the environment $Z_1 := B_2 K_2 Z R_2 R'_2$.
 The state
  $\ket{\tau}^{B_1  K_1 Z_1 R_1R'_1R_2 R'_2}$ 
  has the same reduced state on $B_1 K_1 R_1$ as $\tau$ from
  Eq. (\ref{eq:U0-action}).
  This isometry preserves the fidelity for $\omega_1$, which follows from monotonicity 
  of the fidelity under partial trace:
  \begin{align}
        F&(\sigma_1^{B_1K_1R_1},\tau^{B_1K_1R_1}) \nonumber \\
        & \geq F(\sigma_1^{B_1K_1R_1}\ox \sigma_2^{B_2K_2R_2},\tau^{B_1K_1R_1 B_2K_2R_2}) \nonumber \\
        &\geq 1-\gamma,\nonumber
    \end{align}
  Similarly, we define the isometry 
  $U_2:A_2\hookrightarrow B_1 B_2 K_1 K_2 Z R_1R'_1$
  with output system $B_2 K_2$ and 
  environment $Z_2:=B_1  K_1  Z R_1R'_1$, and the following holds
  \begin{align}
        F&(\sigma_2^{B_2K_2R_2},\tau^{B_2K_2R_2}) \nonumber \\
        & \geq F(\sigma_1^{B_1K_1R_1}\ox \sigma_2^{B_2K_2R_2},\tau^{B_1K_1R_1 B_2K_2R_2}) \nonumber \\
        &\geq 1-\gamma,\nonumber
    \end{align}
  By the above definitions, we obtain
  \begin{align}
     2a(\rho_1 \ox \rho_2,\gamma) &= I(B_1B_2:R_1R'_1 R_2R'_2)_{\tau} \\
     &\geq  I(B_1:R_1R'_1 )_{\tau}+I(B_2:R_2R'_2)_{\tau} \\
     & \geq 2a(\rho_1 ,\gamma) +2a(\rho_2,\gamma).
  \end{align}
where the second line is due to Lemma~\ref{lem:superadditivity-orig}.
The last line follows from the definitions of $a(\rho_1 ,\gamma) $ and $a(\rho_2,\gamma)$.

\item The function is 
  convex for $\gamma \geq 0 $, so it is continuous 
  for $\gamma > 0$. 
  Furthermore, since the function is non-increasing, the convexity implies that it is lower semi-continuous at 
  $\gamma=0$. On the other hand, since the fidelity and the quantum
  mutual information are all continuous functions of CPTP maps, and the domain of the optimization
  is a compact set, the optimum is attained \cite[Thms.~10.1 and 10.2]{Rockafellar1997}, 
  so, the function is also upper 
  semi-continuous at $\gamma=0$.  Combining the two observation, the function is continuous at $\gamma=0$.


\hfill\qedsymbol
\end{enumerate}

\section{Miscellaneous Lemmas and Facts }
\label{sec: Miscellaneous}

\begin{lemma}[{Alicki-Fannes~\cite{Alicki2004}; Winter~\cite{Winter2016}}]
\label{AFW lemma}
Let $\rho$ and $\sigma$ be two states on a bipartite Hilbert space 
$A\otimes B$ with trace distance $\frac12\|\rho-\sigma\|_1 \leq \epsilon$, then
\begin{align*}
  |S(A|B)_{\rho}-S(A|B)_{\sigma}| \leq 2\epsilon \log |A| + (1+\epsilon)h(\frac{\epsilon}{1+\epsilon}).
\end{align*}
\end{lemma}

\noindent The quantum mutual information also satisfies a property called \textit{superadditivity}. 
\begin{lemma}[\cite{Datta2013b}]~\label{lem:superadditivity-orig}
    Let $\rho^{A_1R_1}$ and $\sigma^{A_2R_2}$ be pure quantum states on composite systems $A_1R_1$ and $A_2R_2$. 
    Let $\mathcal{N}^{A_1A_2\to B_1B_2}$ be a quantum channel, and 
    $\omega^{B_1B_2R_1R_2} := \mathcal{N}^{A_1A_2\to B_1B_2}(\rho^{A_1R_1} \otimes\sigma^{A_2R_2})$.
    Then, 
    \begin{align*}
        I(B_1B_2:R_1R_2)_{\omega} \geq I(B_1:R_1)_{\omega} + I(B_2:R_2)_{\omega}. 
    \end{align*}
\end{lemma}

\noindent 
We apply quantum state redistribution~\cite{Devetak2008a,Yard2009} as subprotocol to
construct our direct (achievability) proofs,
which can be summarized as follows.  
\begin{theorem}[Quantum state redistribution~\cite{Devetak2008a,Yard2009}]\label{thm: QSR}
Consider an arbitrary tripartite state on $ACB$, with purification
$|\psi\rangle^{ACBR}$.  Consider $n$ copies of the state for large
$n$, on systems $A_1, \cdots A_n, C_1, \cdots C_n, B_1, \cdots B_n,
R_1, \cdots R_n$.  Suppose initially Alice has systems $A_1, \cdots
A_n, C_1, \cdots C_n$, and Bob has systems $B_1, \cdots B_n$.  Then,
there is a protocol transmitting $Q=n(\frac{1}{2} I(C:R|B) + \eta_n)$
qubits from Alice to Bob, and consuming $nE$ ebits shared between them, 
where $Q+E = S(C|B)$, so that the final state is
$\epsilon_n$-close to $(|\psi\rangle^{ACBR})^{\otimes n}$ but $C_1,
\cdots C_n$ is transmitted from Alice to Bob, and such that
$\{\eta_n\}$, $\{\epsilon_n\}$ are vanishing non-negative sequences.
\end{theorem}


The properties of Koashi-Imoto decomposition are stated in the following theorem.
\begin{theorem}[\cite{Koashi2002,Hayden2004}]
\label{thm: KI decomposition}
Associated to the state $\rho^{AR}$, there are Hilbert spaces $C$, $N$ and $Q$
and an isometry $U_{\KI}:A \hookrightarrow C N Q$ such that:
\begin{enumerate}
  \item The state $\rho^{AR}$ is transformed by $U_{\KI}$ as
    \begin{align}
        \label{eq:KI state}
      (U_{\KI}\!\otimes\! \1_R)\rho^{AR} (U_{\KI}^{\dagger}\!\otimes\! \1_R)
        &\!\!= \!\!\sum_c p_c \proj{c}^{C}\!\! \otimes \omega_c^{N} \otimes \rho_c^{Q R} \nonumber\\
        &=:\omega^{C N Q R},
    \end{align}
    where the set of vectors $\{ \ket{c}^{C}\}$ form an orthonormal basis for Hilbert space 
    $C$, and $p_c$ is a probability distribution over $c$. The states $\omega_c^{N}$ and 
    $\rho_c^{Q R}$ act  on the Hilbert spaces $N$ and $Q \otimes R$, respectively.

  \item For any CPTP map $\Lambda$ acting on system $A$ which leaves the state $\rho^{AR}$ 
    invariant, that is $(\Lambda \otimes \cid_R )\rho^{AR}=\rho^{AR}$, every associated 
    isometric extension $U: A\hookrightarrow A E$ of $\Lambda$ with the environment system 
    $E$ is of the following form
    \begin{equation}
      U = (U_{\KI}\otimes \1_E)^{\dagger}
            \left( \sum_c \proj{c}^{C} \otimes U_c^{N} \otimes \1_c^{Q} \right) U_{\KI},
    \end{equation}
    where the isometries $U_c:N \hookrightarrow N E$ satisfy 
    $\Tr_E [U_c \omega_c U_c^{\dagger}]=\omega_c$ for all $c$.
    The isometry $U_{KI}$ is unique (up to trivial change of basis of the Hilbert spaces 
    $C$, $N$ and $Q$). Henceforth, we call the isometry $U_{\KI}$ and the state 
    $\omega^{C N Q R}=\sum_c p_c \proj{c}^{C} \otimes \omega_c^{N} \otimes \rho_c^{Q R}$ 
    the Koashi-Imoto (KI) isometry and KI-decomposition of the state $\rho^{AR}$, respectively. 

  \item In the particular case of a tripartite system $CNQ$ and a state $\omega^{CNQR}$ already 
    in Koashi-Imoto form (\ref{eq:KI state}), property 2 says the following:
    For any CPTP map $\Lambda$ acting on systems $CNQ$ with 
    $(\Lambda \otimes \cid_R )\omega^{CNQR}=\omega^{CNQR}$, every associated 
    isometric extension $U: CNQ\hookrightarrow CNQ E$ of $\Lambda$ with the environment system 
    $E$ is of the form
    \begin{equation}
      U = \sum_c \proj{c}^{C} \otimes U_c^{N} \otimes \1_c^{Q},
    \end{equation}
    where the isometries $U_c:N \hookrightarrow N E$ satisfy 
    $\Tr_E [U_c \omega_c U_c^{\dagger}]=\omega_c$ for all $c$.
\end{enumerate} 
\end{theorem}
The
sources $\rho^{AR}$ and $\omega^{C N Q R}$  are equivalent in the sense that there are the isometry 
$U_{\KI}$ and the reversal CPTP map $\cR: C N Q \longrightarrow A$, which reverses the 
action of the KI isometry, such that:
\begin{align}
    \omega^{C N Q R}&= (U_{\KI}\otimes \1_R)\rho^{AR} (U_{\KI}^{\dagger}\otimes \1_R), \nonumber \\
    \rho^{AR}&=(\cR \otimes \cid_R)\omega^{C N Q R} \nonumber \\
    &=(U_{\KI}^{\dagger }\otimes \1_R) \omega^{C N Q R} (U_{\KI}\otimes \1_R) \nonumber \Tr [(\1_{C N Q }-\Pi_{C N Q})\omega^{C N Q}]\frac{1}{|A|}\1,
\end{align}
where $\Pi_{C N Q}=U_{\KI}U_{\KI}^{\dagger}$ is the projection onto the subspace 
$U_{\KI}A \subset C \otimes N \otimes Q$. 
We note that both these maps are unital.


\bibliography{References}
\end{document}